\documentclass[preprint,3p,twocolumn]{elsarticle}   % 删除了 'times' 选项（elsarticle 默认使用 Times 字体，无需指定）

\usepackage{amssymb}
\usepackage{amsmath}
\usepackage{amsthm}

\usepackage{xcolor}

\usepackage{subcaption}   % 只保留一次
\usepackage{mathrsfs}
\usepackage{mathdots}
\usepackage{stfloats}     % 可选，提供浮动体选项

\newtheorem{remark}{Remark}

\newtheorem{notation}{Notation}
\newtheorem{theorem}{Theorem}
\usepackage{natbib}
\biboptions{sort&compress}

\journal{International Journal of Systems Science}
\begin{document}

\begin{frontmatter}

\title{Traffic Congestion Control for ARZ Model with an Arbitrarily Large Input Delay \tnoteref{label1}}

\tnotetext[label2]{This work was supported by the National Natural Science Foundation of China under Grants 62373314 and the Research Grants Council of the Hong Kong Special Administrative Region of China under Grant CityU/11210424.}

\author[1]{Yidan Cao}\ead{yidancao4-c@my.cityu.edu.hk}

\author[2]{Xiang Xu}
\ead{xux3@sustech.edu.cn}

\author[1]{Lu Liu\corref{cor1}}
\ead{luliu45@cityu.edu.hk}

\cortext[cor1]{Corresponding author.}

\affiliation[1]{organization={Department of Mechanical Engineering, City University of Hong Kong},
                addressline={Kowloon}, 
                postcode={999077}, 
                state={Hong Kong Special Administrative Region},
                country={China}}

\affiliation[2]{organization={School of Automation and Intelligent Manufacturing, Southern University of Science and Technology},
% addressline={}, 
city={Shenzhen},
postcode={518055}, 
state={Guangdong},
country={China}}

\begin{abstract}
This paper addresses the stabilization problem for Aw-Rascle-Zhang (ARZ) traffic model in the presence of an arbitrarily large input delay. The linearized ARZ model is a $2 \times 2$  hyperbolic partial differential equation (PDE) system with proximal reflection,  which introduces significant analytical challenges when combined with input delays. To tackle this problem, we propose a backstepping-based boundary controller capable of stabilizing the linearized ARZ model under these conditions. 
The input delay is modeled as a transport PDE, which reformulates the entire system into a $3 \times 3$ hyperbolic PDE system.
A backstepping transformation is designed to map the original system into a stable target system, enabling the design of a delay-compensated controller.
A key technical contribution of this work is that for hyperbolic PDEs with delays, we develop a characteristic-region-wise construction  for kernel functions subject to two boundary constraints and close the proof via successive approximation.
Another contribution is that we utilize the small-gain theorem for input-to-state stability (ISS) of hyperbolic PDEs. 
Two simulations are provided to illustrate the effectiveness of the proposed delay-compensated controller: one compares it with a controller without compensation, and the other employs real traffic vehicle data to validate its effectiveness.
\end{abstract}


\end{frontmatter}

\section{Introduction}
Partial differential equations (PDEs) are fundamental tools in many scientific and engineering disciplines due to their ability to model  complex physical phenomena such as fluid flow convection \cite{mojgani2021low}, Euler-Bernoulli beam \cite{mei2023exponential}, flexible beams or plates \cite{avalos2019stability}, electromagnetic or acoustic oscillation \cite{kalman1963mathematical}, heat conduction \cite{colton1977solution}, and traffic flow \cite{yu2019traffic}. Since the 1960s, control of PDE systems has been a challenging research area in control theory \cite{kalman1963mathematical} due to the inherent complexity of infinite-dimensional dynamics.

Among these applications, the modeling and control of traffic flow is a particularly important area, where PDE-based models are essential for traffic congestion problems \cite{richards1956shock, lighthill1955kinematic, payne1971model}. One widely adopted model for traffic flow is the Aw-Rascle-Zhang (ARZ) model \cite{aw2000resurrection,zhang2002non}, a second-order nonlinear hyperbolic PDE system. 
Two primary strategies for traffic control using the ARZ model are ramp metering, which acts as a boundary control input, and variable speed limits. This paper focuses on ramp metering to address traffic congestion problems. Notable works in this domain include 
stabilization problem of linearized ARZ model \cite{yu2019traffic}, output feedback control of two-lane traffic congestion \cite{yu2021output}, traffic flow control on cascaded roads by event-triggered \cite{espitia2022traffic},  neural-operator-based traffic congestion control \cite{zhang2024neural,zhang2025mitigating,lv2025neural}, control and regulation of mixed-autonomy traffic PDE systems \cite{zhang2024mean,zhang2025event,lv2025boundary}, event-triggered control \cite{zhang2025performance,wang2025event} and stabilization problem \cite{wang2025exponential}.
However, a critical limitation of these approaches is their inability to account for time delays, which are prevalent in real-world traffic systems. 

Time delays are a major concern in traffic control systems as they arise from various sources, including signal transmission latency  and computational overhead. These delays can severely degrade the performance of existing controllers and exacerbate the stop-and-go phenomenon, thereby significantly reducing traffic flow efficiency. 
However, although there exist studies on time delays in parabolic PDEs \cite{koudohode2024event}, the significant differences between parabolic and hyperbolic PDE  systems lead to a key issue: for hyperbolic PDEs with time delays, the kernel functions obtained after the backstepping transformation are subject to two boundary conditions, which necessitates a piecewise analysis. This feature is not addressed in  \cite{zhang2023robust}. Moreover, most existing research on time delays in hyperbolic PDEs focuses on delay robustness \cite{auriol2018delay},  
delay-adaptive boundary control \cite{wang2024delay1,wang2024delay2}, and spatially varying delays \cite{qi2024neural,qi2025neural}, while to the best of our knowledge, the case of arbitrarily large delays has not been addressed in the specific setting of non-strict-feedback systems.

In addition to time delays, the ARZ model is characterized as a nonlinear hyperbolic PDE system with proximal reflections (the reflection at the actuated boundary), as introduced in  \cite{auriol2018delay}. However, most existing methodologies, such as \cite{zhang2017necessary, zhang2023robust}, are developed for  strict-feedback hyperbolic PDE systems without boundary reflection terms. The proximal reflection term represents the boundary coupling induced by
the outlet ramp-metering action, as demonstrated in \cite{fan2013comparative}. Although recent progress has been made in developing delay-robust control for strict-feedback hyperbolic PDEs \cite{auriol2019delay} and in handling proximal reflection \cite{auriol2018delay}, the combined effects of input delays and proximal reflection in hyperbolic PDE systems have not been studied. When both proximal reflection and arbitrarily large input delay occur, the PDE system becomes a non-strict feedback system which leads to difficulties in designing controllers and the presence of proximal reflection introduces strong coupling terms that render conventional Lyapunov-based analysis intractable.

Consequently, these gaps motivate the present work.
We establish a backstepping transformation for controller design and introduce the input-to-state stability (ISS) method for closed-loop system stability analysis.
Originally proposed in \cite{sontag1989smooth} for finite-dimensional systems, ISS theory was later extended to PDEs \cite{mironchenko2015construction}, enabling the small-gain theorem to be applied to interconnected PDE systems \cite{mironchenko2017characterizations}.
Recent advances have further generalized the small-gain theorem to accommodate PDEs with boundary disturbances \cite{karafyllis2016iss}, and non-local boundary conditions \cite{karafyllis2017iss}.

This paper builds upon these developments to address the stabilization problem for traffic flow systems governed by the ARZ model with arbitrarily large input delay. We propose a backstepping-based boundary control framework to stabilize a $2 \times 2$ hyperbolic PDE system with proximal reflection and arbitrarily large input delays. 
The main contributions of this work can be summarized as follows:

 First, a novel well-posedness design for the delay-compensating kernel system arising in the hyperbolic PDEs is established.
 Unlike standard single-boundary settings, the kernel equations here exhibit a two-boundary  structure induced by delayed boundary actuation. We therefore develop a characteristic-region-wise construction for kernel functions subject to two boundary constraints and close the proof via successive approximation, thereby obtaining the well-posedness of the kernel solution.

Second, a delay-compensated boundary controller is designed. 
While \cite{yu2019traffic} does not account for input delays and
\cite{zhang2023robust} does not address the PDE systems with
proximal reflection, this work simultaneously considers input delay and proximal reflection in the ARZ ramp-metering setting.
To exponentially stabilize the closed-loop system, we construct a delay-compensating backstepping transformation and design a delay-compensated boundary controller.

Third, the closed-loop stability is established using the ISS small-gain
approach. Although the construction of Lyapunov functions is a conventional methodology, it is particularly challenging for the $2 \times 2$ hyperbolic PDEs with proximal reflections, because of the strong coupling caused by proximal reflection and the problems brought about by time-delay. To address this limitation, we employ an ISS small-gain approach to establish stability. Compared to conventional Lyapunov-based methods, the small-gain framework offers greater flexibility for handling the strong coupling effects such as the effect of proximal reflection. This extension leads to a stability criterion for the closed-loop system.

The rest of this paper is organized as follows: Section \ref{2} presents the problem formulation and introduces the transformation of input delay to a transport PDE. In Section \ref{3}, we develop a backstepping-based compensation controller, establish the well-posedness of kernel functions and provide a rigorous stability analysis. Section \ref{4} provides numerical simulations to demonstrate the effectiveness of the proposed controller. Finally, conclusions and perspectives are given in Section \ref{5}.

\begin{notation}
We denote  $\mathcal{L}^{\infty}(a,b)$  the space of bounded real-valued functions defined on  $[a,b]$  with the standard  $\mathcal{L}^{\infty}$  norm, i.e., for any $f \in \mathcal{L}^{\infty}(a,b)$
\begin{align}
&\|f\|_{L^{\infty}}=\sup _{x \in[a,b]}\|f(x)\|.
\end{align}

The exponentially weighted norm is
\begin{align}
\|h\|_{\mu,t}
=
\sup_{0\leq s\leq t}
e^{\mu s}\|h(\cdot,s)\|_{L^\infty(0,L)},
\qquad \mu>0.
\end{align}
Hence,
\begin{align}
\|h(\cdot,t)\|_{L^\infty(0,L)}
\leq
e^{-\mu t}\|h\|_{\mu,t}.
\end{align}
\end{notation}

\section{Problem Formulation}\label{2}
In this section, we first introduce the ARZ traffic model and describe how an arbitrarily large input delay can be represented as a transport PDE. We then transform the original system to an intermediate system that facilitates the design of the target system.

\begin{figure}
    \centering
    \includegraphics[width=0.5\linewidth]{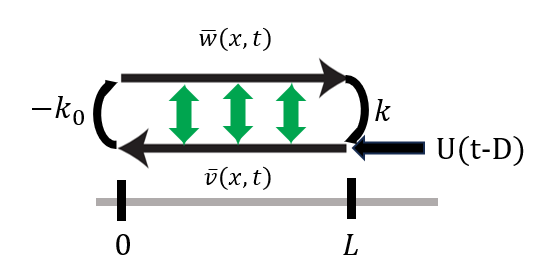}
    \caption{System diagram of original time-delay system}
    \label{fig1}
\end{figure}

\begin{figure}
    \centering
    \includegraphics[width=0.7\linewidth]{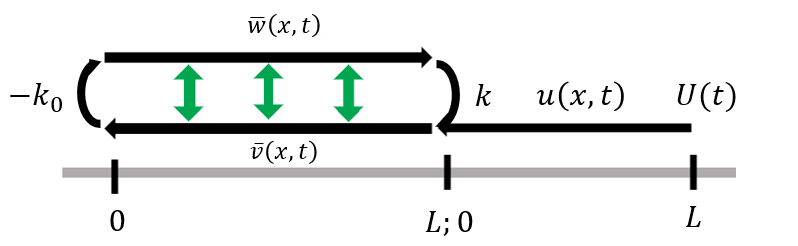}
    \caption{System diagram after converting time delay into transport PDEs}
    \label{fig2}
\end{figure}

\subsection{ARZ model}

The macroscopic behavior of traffic along a roadway is governed by the ARZ model \citep{aw2000resurrection, zhang2002non}, which consists of a set of nonlinear, coupled hyperbolic PDEs. This model describes the spatiotemporal evolution of traffic dynamics in the domain $(x, t) \in[0, L] \times[0,+\infty)$. The model is defined by
\begin{align}
&\partial_{t} \rho+\partial_{x}(\rho v) =0, \label{arz1}\\
&\partial_{t}(v-V(\rho))+v \partial_{x}(v-V(\rho))  =\frac{V(\rho)-v}{\tau},\label{arz2}
\end{align}
where $V(\rho)$ is the equilibrium traffic speed profile, $\rho(x, t)$ represents the traffic density and $v(x, t)$ represents the traffic speed. The parameter $\tau \in \mathbb{R}>0$ captures the driver's reaction time. To ensure realistic behavior, the flow function $q(\rho)=\rho V(\rho)$ is required to be strictly concave, satisfying $q^{\prime \prime}(\rho)<0$ \citep{zhang2025mitigating}. The equilibrium velocity-density relationship $V(\rho)$ is given in the form of Greenshield's model \citep{greenshields1935study}, given by
\begin{align}
V(\rho)=v_{f}-p(\rho)=v_{f}\left(1-\left(\frac{\rho}{\rho_{m}}\right)^{\gamma}\right),
\end{align}
where  $\gamma \in \mathbb{R}_+$, $v_{f}$  is the maximum speed for the traffic flow,  $\rho_{m}$  denotes the maximum density and $p(\rho)$ is defined as the traffic pressure, an increasing function of density:
\begin{align}
    p(\rho)=v_{f}\left(\frac{\rho}{\rho _m}\right)^\gamma,
\end{align}
where $\rho_m$ is the maximum density, thus $V(\rho_m)=0$.
The traffic pressure  $p(\rho)$  and flux  $q$  are related by
\begin{align}
p=\frac{v_{f}}{{\rho_{m}}^{\gamma}}\left(\frac{q}{v}\right)^{\gamma} .
\end{align}

The system's equilibrium, defined as the steady state $(\rho^{\star}, v^{\star})$, is derived from this fundamental diagram, establishing the necessary relation between the equilibrium speed and density. 
\begin{align}
v^{\star}=V\left(\rho^{\star}\right).
\end{align}
For the freeway traffic under consideration, the inlet boundary at $x=0$ is maintained at a constant equilibrium flow, defined as $q^{\star}=\rho^{\star} v^{\star}$. This leads to the following inlet boundary condition:
\begin{align}
\rho(0, t)=\frac{q^{\star}}{v(0, t)}.\label{rho}
\end{align}
At the outlet of the road section, we set the traffic density as  $\rho^{\star} $ to obtain the following boundary condition for traffic speed
\begin{align}
v(L, t)=\frac{q(L, t)}{\rho^{\star}}+U(t),\label{boundv}
\end{align}
where $U(t)$ denotes the control input, which manipulates the outflow from the road section. 

Our control strategy targets the traffic flow upstream of a ramp meter (UORM), where the control input $U(t)$ is applied physically at the downstream outlet of the domain via ramp metering to regulate the discharge flow. This setup yields a closed-loop system in which upstream measurements influence downstream actuation. The objective is to drive both traffic density and speed to their equilibrium values $(\rho^{\star}, v^{\star})$.

\subsection{Linearizing the ARZ model}
The PDE system \eqref{arz1}-\eqref{arz2} is first linearized around its equilibrium point $(q^{\star}, v^{\star})$. The deviations from this equilibrium are defined as:
\begin{align}
\tilde{q}(x, t)=q(x, t)-q^{\star},\\
\tilde{v}(x, t)=v(x, t)-v^{\star}.
\end{align}
To enable the subsequent backstepping-based boundary control design, the linearized system is first reformulated into a target boundary control model. 
We define new variables $(\bar{w}, \bar{v})$ in Riemann coordinates,
\begin{align}
&\bar{w}=e^{\left(\frac{x}{\tau v^{\star}}\right)}\left(\tilde{q}-q^{\star}\left(\frac{1}{v^{\star}}-\frac{1}{\gamma p^{\star}}\right) \tilde{v}\right), \\
&\bar{v}=\frac{q^{\star}}{\gamma p^{\star}} \tilde{v}.
\end{align}

After linearizing the original system \eqref{arz1}-\eqref{arz2} with boundary conditions \eqref{rho}-\eqref{boundv} at the steady state $(p^{\star},v^{\star})$, we can obtain a linear $2\times 2$ hyperbolic PDE described by the state variables $\bar{w}(x, t)$ and $\bar{v}(x, t)$. We consider the following system
\begin{align}
\bar{w}_{t}(x, t)&=-v^{\star} \bar{w}_{x}(x, t), \label{orin} \\
\bar{v}_{t}(x, t)&=\left(\gamma p^{\star}-v^{\star}\right) \bar{v}_{x}(x, t)+c(x) \bar{w}(x, t),\label{orin1}\\ 
\bar{w}(0, t) &=-k_{0} \bar{v}(0, t), \label{bond1}\\
\bar{v}(L, t) &=k \bar{w}(L, t)+U(t-D),\label{output}
\end{align}
where $D>0$ denotes an arbitrarily large input delay, with $v^{\star}>0,\left(\gamma p^{\star}-v^{\star}\right)>0$ and
\begin{align}
c(x)&=-\frac{1}{\tau} \exp \left(-\frac{x}{\tau v^{\star}}\right),\\
k_{0} &=\frac{\gamma p^{\star}-v^{\star}}{v^{\star}},\label{k0}  \\
k &=e^{ \left(-\frac{L}{\tau v^{\star}}\right)}.\label{k}
\end{align}
Since the model represents a section of highway, we assume $L> 0$.

\subsection{Arbitrarily large input delay }

Based on the framework in \cite{krstic2009delay}, we transform the arbitrarily large input delay into a transport equation, which is a first-order hyperbolic PDE.
This transformation maps the system in Fig. \ref{fig1} to the equivalent block diagram in Fig. \ref{fig2}.

The input delay can be transformed into the following PDE
\begin{align}
U(t) &=u(L,t),\label{u}\\
\frac{D}{L}u_{t} &= u_{x}.
\end{align}
Therefore, the control input with arbitrarily large input delay can be written as
\begin{align}
u(0,t) &=U(t-D).\label{ubond}
\end{align}
Thus, the system \eqref{orin}-\eqref{output} can be rewritten as follows
\begin{align}
\bar{w}_{t}(x, t) &=-v^{\star} \bar{w}_{x}(x, t), \label{origin1}\\
\bar{v}_{t}(x, t) &=\left(\gamma p^{\star}-v^{\star}\right) \bar{v}_{x}(x, t)+c(x) \bar{w}(x, t), \\
\bar{w}(0, t) &=-k_{0} \bar{v}(0, t), \\
\bar{v}(L, t) &=k \bar{w}(L, t)+u(0,t),\\
\frac{D}{L}u_{t} &= u_{x},\\
u(L,t) &= U(t).\label{origin2}
\end{align}

\section{Main results} \label{3}

\subsection{Backstepping-based controller design}
We introduce three integral transformations to map the original system \eqref{origin1}-\eqref{origin2} to a target system \eqref{target1}-\eqref{target2}. The transformations can be constructed as follows
\begin{align}
\alpha \left( x,t \right) &=\bar{w} \left( x,t \right ), \label{map1}\\ 
\beta  \left ( x,t \right ) &=\bar{v} \left ( x,t \right ) - \int_{0}^{x} K^{11} \left ( x,y \right ) \bar{w}\left ( y,t \right )dy  \notag \\
&-\int_{0}^{x}K^{12}\left ( x,y \right ) \bar{v} \left ( y,t \right )dy,\label{map2} \\
z \left ( x,t \right ) &=u \left ( x,t \right )- \int_{0}^{x} P(x,y) u\left ( y,t \right )dy \notag \\
&-\int_{0}^{L} M \left ( x,y \right ) \bar{w}\left ( y,t \right )dy\notag \\
&- \int_{0}^{L}N\left ( x,y \right ) \bar{v} \left ( y,t \right )dy,\label{map3} 
\end{align}
where $(x,t)\in[0,L]\times[0,+\infty)$,
the kernels $K^{1j}$, $j=1,2$, and $P$ are defined on the triangular
domain
$\mathcal T_1=\{(x,y)\in[0,L]^2:\ 0\le y\le x\le L\}$,
whereas $M$ and $N$ are defined on the rectangular domain
$\mathcal T_2=\{(x,y)\in[0,L]^2:\ 0\le x\le L,\ 0\le y\le L\}$.

According to \eqref{map1}-\eqref{map3}, the original system \eqref{origin1}-\eqref{origin2} can be transformed into the following target system
\begin{align}
{\alpha }_{t}(x, t) & =-v^{\star} {\alpha }_{x}(x, t), \label{target1}\\ 
{\beta }_{t}(x, t) & =\left(\gamma p^{\star}-v^{\star}\right) {\beta }_{x}(x, t),\label{targetv} \\
{\alpha}(0, t) & =-k_{0} {\beta }(0, t),\label{targetbound1} \\
{\beta }(L, t) &=k{\alpha}(L, t) + z(0,t),\label{targetbound}\\
\frac{D}{L}z_{t} &= z_{x},\label{targetz}\\
z(L,t) &=U(t)-\int_{0}^{L} P(L,y) u(y, t) \mathrm{d} y \notag \\ 
&-\int_{0}^{L} M(L, y) \bar{w}(y, t) \mathrm{d} y \notag \\ 
&-\int_{0}^{L} N(L, y) \bar{v}(y, t) \mathrm{d} y .\label{target2}
\end{align}
By differentiating the mappings \eqref{map1}-\eqref{map2} with respect to time and domain, applying integration by parts and incorporating boundary conditions, the resulting constraints on the kernel functions guarantee a proper mapping from the original system to the target system. 
The kernel functions must satisfy
\begin{align}
&\left(\gamma p^{\star}-v^{\star}\right)K_{x}^{11}(x,y)-v^{\star}K_{y}^{11}(x,y)\notag \\
&-c(y)K^{12}(x,y) =0, \label{k1}\\
&K_{y}^{12}(x,y)+K_{x}^{12}(x,y) =0,\label{k2}
\end{align}
with the boundary conditions given by
\begin{align}
c(x)+\gamma p^{\star}K^{11}(x,x) &=0,\\
(\gamma p^{\star}-v^{\star})K^{12}(x,0)+k_{0}v^{\star}K^{11}(x,0) &=0.
\end{align}
After simplification, we obtain
\begin{align}
K^{11}(x,x) &=-\frac{c(x)}{\gamma p^{\star}},\\
K^{12}(x,0) &=-\frac{k_{0}v^{\star}}{\gamma p^{\star}-v^{\star}}K^{11}(x,0). 
\end{align}
Notice that \eqref{k1}-\eqref{k2} can be written as two separate $2 \times 2$ hyperbolic systems involving $K^{11}$ and $K^{12}$.
The well-posedness of the kernel equations $K^{1j}, j=1,2$ is established using an approach analogous to the Appendix of \cite{vazquez2011backstepping}.

Differentiating the mapping  \eqref{map3} with respect to space and time yields the following equations for $M(x,y),N(x,y)$ 
\begin{align}
\frac{D}{L} v^{\star}M_{y}(x,y)-M_{x}(x,y)+\frac{D}{L}N(x,y)c(y) &=0, \label{Mxy}\\
\frac{D}{L}(\gamma p^{\star}-v^{\star})N_{y}(x,y)+N_{x}(x,y) &=0,\label{Nxy}
\end{align}
with boundary conditions
\begin{align}
\frac{D}{L}\left(\gamma p^{\star}-v^{\star}\right)N(x,L) &=\frac{Dv^{\star}}{kL}M(x,L),\label{M} \\
k_{0}v^{\star}M(x,0) &=-\left(\gamma p^{\star}-v^{\star}\right)N(x,0),\label{N}\\
M(0,y)&=K^{11}(L,y), 
\label{Mleft}\\
N(0,y)&=K^{12}(L,y). \label{Nleft}
\end{align}

It can be further obtained that the kernel \(P(x,y)\) satisfies
\begin{align}
P_x(x,y)+P_y(x,y)&=0,\qquad (x,y)\in\mathcal T_1,
\label{Pxy}\\
P(x,0)&=\frac{D v^\star}{kL}M(x,L),\qquad x\in[0,L].
\label{P0}
\end{align}

Solving \eqref{Pxy} along the characteristic curves, and using \eqref{P0}, we obtain, for \((x,y)\in\mathcal T_1\),
\begin{align}
P(x,y)
=
P(x-y,0)
=
\frac{Dv^\star}{kL}M(x-y,L).
\label{Pexplicit}
\end{align}

Therefore, the well-posedness and boundedness of \(P(x,y)\) follow from
those of \(M(x,y)\). The well-posedness of  $M(x,y)$  and  $N(x,y)$ can be proven according to the following theorem.

\begin{theorem}\label{th3}
The kernel equations \eqref{Mxy}--\eqref{Nleft} have a unique bounded
solution in the characteristic sense, 
\[
M(x,y),\ N(x,y)\in L^\infty([0,L]\times[0,L]).
\]
Moreover, there exist constants \(\bar M,\bar N>0\) such that
\[
|M(x,y)|\leq \bar M,\qquad |N(x,y)|\leq \bar N,
\]
for  \((x,y)\in[0,L]\times[0,L]\).
\end{theorem}

\begin{proof}
We begin by solving for \(N(x,y)\). Define
$a=\frac{D}{L}\left(\gamma p^\star-v^\star\right)$.
Equation \eqref{Nxy} becomes
\begin{align}
aN_y+N_x=0 . \label{aN}
\end{align}
The characteristic equation is
\begin{align}
\frac{dx}{1}=\frac{dy}{a},
\end{align}
which gives
\begin{align}
y-ax=\xi .
\end{align}
Thus,
\begin{align}
N(x,y)=f(y-ax), \label{Nf}
\end{align}
where \(f\) is to be determined.

From the boundary condition \eqref{Nleft},
\begin{align}
f(y)=K^{12}(L,y),\qquad 0\leq y\leq L. \label{fpositive}
\end{align}
Hence it remains to determine \(f(\xi)\) for \(-aL\leq \xi<0\).

Define
\[
b=\frac{D}{L}v^\star,\qquad d=\frac{D}{L}.
\]
Substituting \eqref{Nf} into \eqref{Mxy} gives
\begin{align}
bM_y-M_x=-dc(y)f(y-ax). \label{Msub}
\end{align}
The characteristic equation is
\begin{align}
\frac{dx}{1}=\frac{dy}{-b},
\end{align}
which gives
\begin{align}
y+bx=\eta .
\end{align}
Along the characteristic \(y(s)=y+b(x-s)\), we can obtain
\begin{align}
\frac{d}{ds}M(s,y+b(x-s))
=&dc(y+b(x-s))\notag \\
&f(y+bx-(a+b)s). \label{Mode}
\end{align}

If \(y+bx\leq L\), the characteristic starts from \(x=0\). Using
\eqref{Mleft}, we obtain
\begin{align}
M(x,y)
=&K^{11}(L,y+bx)
+d\int_0^x c(y+b(x-s)) \notag\\
&\times f(y+bx-(a+b)s)\,ds . \label{Mcase1}
\end{align}

If \(y+bx>L\), the characteristic starts from \(y=L\). Let
$\sigma=x-\frac{L-y}{b}$.
From \eqref{M}, we can obtain
\begin{align}
aN(\sigma,L)=\frac{Dv^\star}{kL}M(\sigma,L)
=\frac{b}{k}M(\sigma,L).
\end{align}
Therefore,
\begin{align}
M(\sigma,L)
=\frac{ka}{b}N(\sigma,L)
=\frac{ka}{b}f(L-a\sigma).
\end{align}
Thus,
\begin{align}
M(x,y)
=&\frac{ka}{b}
f\left(L-a\left(x-\frac{L-y}{b}\right)\right) \notag\\
&+d\int_{x-\frac{L-y}{b}}^x c(y+b(x-s)) \notag\\
&\times f(y+bx-(a+b)s)\,ds . \label{Mcase2}
\end{align}

Now we set \(y=0\). If \(bx\leq L\), then \eqref{Mcase1} gives
\begin{align}
M(x,0)
=&K^{11}(L,bx)
+d\int_0^x c(b(x-s)) \notag\\
&\times f(bx-(a+b)s)\,ds . \label{Mzero1}
\end{align}
Using \eqref{N}, it can be further implied that
\begin{align}
k_0v^\star M(x,0)
=-\left(\gamma p^\star-v^\star\right)f(-ax).
\end{align}
Hence,
\begin{align}
f(-ax)
=&-\frac{k_0v^\star}{\gamma p^\star-v^\star}
\bigg[
K^{11}(L,bx) \notag\\
&+d\int_0^x c(b(x-s))
f(bx-(a+b)s)\,ds
\bigg],\notag \\
\label{Vol1}
\end{align}
where $bx\leq L$.
If \(bx>L\), then \eqref{Mcase2} gives
\begin{align}
M(x,0)
=&\frac{ka}{b}f\left(L-a\left(x-\frac{L}{b}\right)\right) \notag\\
&+d\int_{x-\frac{L}{b}}^x c(b(x-s))
f(bx-(a+b)s)\,ds . \label{Mzero2}
\end{align}
Therefore,
\begin{align}
f(-ax)
=&-\frac{k_0v^\star}{\gamma p^\star-v^\star}
\bigg[
\frac{ka}{b}f\left(L-a\left(x-\frac{L}{b}\right)\right) \notag\\
&+d\int_{x-\frac{L}{b}}^x c(b(x-s))
f(bx-(a+b)s)\,ds
\bigg],
 \label{Vol2}
\end{align}
where $bx>L$.
Equations \eqref{Vol1}--\eqref{Vol2}, together with \eqref{fpositive},
form a Volterra-type integral equation for \(f\) on \((-aL,L)\).
The values of \(f\) on \((0,L)\) are prescribed by \eqref{fpositive},
whereas the values of \(f\) on \((-aL,0)\) are unknown.

Choose 
% \(f^0\in (-aL,L)\) such that
\begin{align}
f^0(y)=K^{12}(L,y),\qquad 0<y<L.
\end{align}

For \(n=0,1,2,\ldots\), set
\begin{align}
f^{n+1}(y)=K^{12}(L,y),\qquad 0<y<L.
\label{iter0}
\end{align}
For \(0<x<L\) with \(bx\le L\), define
\begin{align}
f^{n+1}(-ax)
=&-\frac{k_0v^\star}{\gamma p^\star-v^\star}
\bigg[
K^{11}(L,bx) \notag\\
&+d\int_0^x c(b(x-s))
f^n(bx-(a+b)s)\,ds
\bigg].
\label{iter1}
\end{align}
For \(0<x<L\) with \(bx>L\), define
\begin{align}
f^{n+1}(-ax)
=&-\frac{k_0v^\star}{\gamma p^\star-v^\star}
\bigg[
\frac{ka}{b}
f^n\left(L-a\left(x-\frac{L}{b}\right)\right) \notag\\
&+d\int_{x-\frac{L}{b}}^x c(b(x-s))
f^n(bx-(a+b)s)\,ds
\bigg].
\label{iter2}
\end{align}

For \(\rho>0\), define
\begin{align}
\|f\|_\rho
=
\operatorname*{sup}_{0<x<L}
e^{-\rho x}|f(-ax)|.
\end{align}

For \(0<x<L\) with \(bx<L\), subtracting two consecutive iterations gives
\begin{align}
&f^{n+1}(-ax)-f^n(-ax) \notag\\
=&-\frac{k_0v^\star d}{\gamma p^\star-v^\star}
\int_0^x c(b(x-s))
\Big[
f^n(bx-(a+b)s) \notag\\
&-f^{n-1}(bx-(a+b)s)
\Big]\,ds .
\label{diff1}
\end{align}
For \(0<x<L\) with \(bx>L\), one obtains
\begin{align}
&f^{n+1}(-ax)-f^n(-ax) \notag\\
=&-\frac{k_0v^\star}{\gamma p^\star-v^\star}
\bigg[
\frac{ka}{b}
\Big(
f^n\left(L-a\left(x-\frac{L}{b}\right)\right) \notag\\
&-f^{n-1}\left(L-a\left(x-\frac{L}{b}\right)\right)
\Big) \notag\\
&+d\int_{x-\frac{L}{b}}^x c(b(x-s))
\Big[
f^n(bx-(a+b)s) \notag\\
&-f^{n-1}(bx-(a+b)s)
\Big]\,ds
\bigg].
\label{diff2}
\end{align}

Since \(f^n(y)=f^{n-1}(y)=K^{12}(L,y)\) for \(0<y<L\), the difference vanishes whenever the argument belongs to \((0,L)\).

For the integral terms, if
\[
bx-(a+b)s<0,
\]
then
\[
bx-(a+b)s=-a\left(s-\frac{b(x-s)}{a}\right),
\]
and
\[
x-\left(s-\frac{b(x-s)}{a}\right)
=
\frac{a+b}{a}(x-s).
\]
Therefore,
\begin{align}
&e^{-\rho x}
\left|
f^n(bx-(a+b)s)-f^{n-1}(bx-(a+b)s)
\right| \notag\\
&\leq
e^{-\rho\frac{a+b}{a}(x-s)}
\|f^n-f^{n-1}\|_\rho .
\end{align}
Thus,
\begin{align}
d\|c\|_{L^\infty}
\int_0^x
e^{-\rho\frac{a+b}{a}(x-s)}\,ds
\leq
\frac{d\|c\|_{L^\infty}a}{\rho(a+b)} .
\end{align}

For
$L-a\left(x-\frac{L}{b}\right)
=
-a\left(x-\frac{L}{b}-\frac{L}{a}\right)$,
hence
\begin{align}
&e^{-\rho x}
\left|
f^n\left(L-a\left(x-\frac{L}{b}\right)\right)
-
f^{n-1}\left(L-a\left(x-\frac{L}{b}\right)\right)
\right| \notag\\
&\leq
e^{-\rho\left(\frac{L}{a}+\frac{L}{b}\right)}
\|f^n-f^{n-1}\|_\rho .
\end{align}

Combining the above estimates gives
\begin{align}
\|f^{n+1}-f^n\|_\rho
\leq
\kappa_\rho
\|f^n-f^{n-1}\|_\rho,
\end{align}
where
\begin{align}
\kappa_\rho
=&
\left|
\frac{k_0v^\star}{\gamma p^\star-v^\star}
\frac{ka}{b}
\right|
e^{-\rho\left(\frac{L}{a}+\frac{L}{b}\right)}
\notag\\
&+
\frac{
d|k_0|v^\star \|c\|_{L^\infty}a
}{
(\gamma p^\star-v^\star)\rho(a+b)
}.
\end{align}
Since
\begin{align}
\lim_{\rho\to+\infty}\kappa_\rho=0,
\end{align}
one can choose \(\rho>0\) such that
\begin{align}
\kappa_\rho<1.
\end{align}
Therefore, \(\{f^n\}\) is a Cauchy sequence in \(L^\infty(-aL,0)\), the same contraction estimate applied to two fixed points gives uniqueness.

Therefore, \(N(x,y)\) is uniquely determined by \eqref{Nf}. Then \(M(x,y)\)
is uniquely determined by \eqref{Mcase1}--\eqref{Mcase2}. Hence the kernel
systems \eqref{Mxy}--\eqref{Nleft} have a unique bounded characteristic
solution.

Finally,
\begin{align}
\|N\|_{L^\infty}
\leq
\|f\|_{L^\infty(-aL,L)}
=\bar N .
\end{align}
From \eqref{Mcase1}--\eqref{Mcase2},
\begin{align}
\|M\|_{L^\infty}
&\leq
\max\left\{
\|K^{11}(L,\cdot)\|_{L^\infty},
\left|\frac{ka}{b}\right|\bar N
\right\}
+dL\|c\|_{L^\infty}\bar N \notag \\
&=\bar M .
\end{align}
The proof is complete.
\end{proof}

After obtaining all the kernel functions, we derive the control law 
\begin{align} \label{controller1}
U(t) &= \int_{0}^{L} P(L,y) u(y, t) \mathrm{d} y+\int_{0}^{L} M(L, y) \bar{w}(y, t) \mathrm{d} y \notag \\ 
&+\int_{0}^{L} N(L, y) \bar{v}(y, t) \mathrm{d} y .
\end{align}

Here the design  controller  variables are the kernels $M(x,y), N(x,y), P(x,y)$, obtained sequentially from \eqref{Mxy}-\eqref{Pexplicit}. Hence, once these kernel equations are solved, the controller \eqref{controller1} is fully specified.

Thus, the target system is given by:
\begin{align}
{\alpha }_{t}(x, t) & =-v^{\star} {\alpha }_{x}(x, t), \\ 
{\beta }_{t}(x, t) & =\left(\gamma p^{\star}-v^{\star}\right) {\beta }_{x}(x, t), \\
{\alpha}(0, t) & =-k_{0} {\beta }(0, t), \\
{\beta }(L, t) &=k{\alpha}(L, t) + z(0,t),\\
\frac{D}{L}z_{t} &= z_{x},\\
z(L,t) &=0 .
\end{align}

\subsection{The inverse transformations}
In this subsection, we construct inversion mappings from $(\alpha, \beta, z)$ to $(\bar{w},\bar{v},u)$, which converts the target system \eqref{target1}-\eqref{target2} back to the original system \eqref{origin1}-\eqref{origin2},  which is fundamental for the stability of our control design.

The inverse transformation functions are given by
\begin{align}
\bar{w}  \left ( x,t \right ) &=\alpha \left ( x,t \right ), \label{invermap1} \\
\bar{v} \left ( x,t \right ) &=\beta \left ( x,t \right ) + \int_{0}^{x} L^{11} \left ( x,y \right ) \alpha\left ( y,t \right )dy \notag \\
&+ \int_{0}^{x}L^{12}\left ( x,y \right ) \beta \left ( y,t \right )dy, \label{invermap2}\\
u \left ( x,t \right ) &=z \left ( x,t \right )+ \int_{0}^{x} H(x,y)  z\left ( y,t \right )dy \notag \\
&+ \int_{0}^{L} \Phi   \left ( x,y \right ) \alpha\left ( y,t \right )dy\notag \\
&+ \int_{0}^{L}\Psi \left ( x,y \right ) \beta \left ( y,t \right )dy, \label{invermap3}
\end{align}
where the kernel functions \(L^{11}(x,y)\) and \(L^{12}(x,y)\) are defined on \(\mathcal T_1\). The kernel functions $\Phi  (x, y)$, $\Psi (x, y)$  are defined on the rectangular domain  $\mathcal{T}_{2}=\left\{(x, y) \in \mathbb{R}^{2}: 0 \leq x, y \leq L\right\}$, and $H\in C^1(\mathcal T_1),
\mathcal T_1=\{(x,y)\in[0,L]^2:0\le y\le x\le L\}.$

Differentiating the mappings \eqref{invermap1}-\eqref{invermap2} with respect to space and time, performing integration by parts and applying the boundary conditions, the kernel functions $L^{1j}(x, y)$, $j=1,2$ should satisfy
\begin{align}
&\left(\gamma p^{\star}-v^{\star}\right)L_{x}^{11}(x,y)-v^{\star}L_{y}^{11}(x,y)=0,\\
&\left(\gamma p^{\star}-v^{\star}\right)L_{x}^{12}(x,y)+\left(\gamma p^{\star}-v^{\star}\right)L_{y}^{12}(x,y) =0.
\end{align}
The boundary conditions of the kernel functions $L^{1j}(x,y)$, $j=1,2$ are given as follows
\begin{align}
L^{11}(x,x) &=-\frac{c(x)}{\gamma p^{\star}}, \\
L^{12}(x,0) &=-\frac{k_{0}v^{\star}}{\gamma p^{\star}-v^{\star}}L^{11}(x,0) .
\end{align}
Notice that they can be written as two separate $2 \times 2$ hyperbolic systems for $L^{11}$ and $L^{12}$. The well-posedness of the two kernels can be proved according to \cite{vazquez2011backstepping}.

Then by differentiating the mapping \eqref{invermap3} with respect to space and time, and let $\varrho =\gamma p^\star-v^\star,  \vartheta=\frac{D}{L}$. 
Then the inverse kernels \(\Phi  (x,y)\), \(\Psi (x,y)\), and \(H(x,y)\) should satisfy
\begin{align}
\Psi _x(x,y)+\vartheta \varrho  \Psi _y(x,y)&=0,
\label{invBxy}\\
\Phi  _x(x,y)-\vartheta v^\star \Phi  _y(x,y)&=0,
\label{invAxy}\\
H_x(x,y)+H_y(x,y)&=0,
\label{invHxy}
\end{align}
with boundary conditions
\begin{align}
v^\star \Phi  (x,L)&=k\varrho  \Psi (x,L),
\label{invABtop}\\
H(x,0)&=\vartheta \varrho  \Psi (x,L),
\label{invH0}\\
\Phi(0,y)&=L^{11}(L,y), \label{eq:Phi_left_bc}\\
\Psi(0,y)&=L^{12}(L,y),\\
k_0v^\star \Phi  (x,0)+\varrho  \Psi (x,0)&=0.
\label{invABbottom}
\end{align}

The well-posedness of \eqref{invermap1}–\eqref{invermap3} can be proved using a method similar to that used in the proof of Theorem \ref{th3} above. Since the inverse kernel equations have the same transport structure and bounded boundary data, their bounded characteristic solutions follow by the same characteristic argument.

Hence, the controller can also be written as
\begin{align}
    U(t)&= \int_{0}^{L} H(L,y)  z\left ( y,t \right )dy + \int_{0}^{L} \Phi   \left ( L,y \right ) \alpha\left ( y,t \right )dy \notag\\
    &+ \int_{0}^{L}\Psi \left ( L,y \right ) \beta \left ( y,t \right )dy. \label{controller2}
\end{align}

\subsection{Stability analysis}
In this subsection, we provide a stability analysis of the original closed-loop system \eqref{orin}-\eqref{output} under the controller \eqref{controller1}. In the previous subsection, we have performed inverse transformations and established the equivalence between the target system \eqref{target1}-\eqref{target2} and the original closed-loop system. These developments  provide the basis for   stability analysis of the closed-loop systems.

More specifically, after converting the input delay into a transport PDE and applying the backstepping transformation, the closed-loop target system can be viewed as an interconnection of boundary-coupled transport subsystems. 
The ISS small-gain approach is particularly suitable here because it exploits the boundary-interconnection structure of the $\alpha-$ and $\beta-$ transport subsystems, represents the proximal reflection and delay-state couplings through subsystem gains, and reduces the stability analysis to the explicit and verifiable condition, which directly yields exponential stability of the interconnected target system. In contrast, a direct Lyapunov analysis generates boundary trace cross terms such as $\alpha(L,t)z(0,t)$, thereby increasing conservatism and making the resulting stability criterion substantially less transparent.

To begin with, we analyze the stability of the target system \eqref{target1}-\eqref{target2} under the controller \eqref{controller1}.

\begin{theorem}
Consider the closed-loop system consisting of system dynamics and boundary conditions \eqref{target1}-\eqref{target2} under controller \eqref{controller1}. If  $|kk_0|<1$, then for any initial conditions $\alpha(\cdot,0) ,\beta(\cdot,0),z(\cdot,0)\in \mathcal{L} ^{\infty}(0,L)$, the closed-loop system is exponentially stable in the sense that there exist positive constants $\tilde{C}$ and $\lambda$ such that
\begin{align}
\Omega (t)\leq \tilde{C} e^{-\lambda t}\Omega (0),
\end{align}
where $t >0$ and 
\begin{align}
\Omega (t)=\|\alpha(\cdot,t) \|_{L^\infty}+\|\beta(\cdot,t) \|_{L^\infty}+\|z(\cdot,t) \|_{L^\infty}.
\end{align}
\end{theorem}

\begin{proof}
Based on the target system \eqref{target1}-\eqref{target2}, we have
\begin{align}
\alpha_{t}+v^{\star} \alpha_{x}&=0,\\
\beta_{t}&=\varrho   \beta_{x},
\end{align}
where $\varrho  =\gamma p^{\star}-v^{\star}$. By using method of characteristics, we can obtain the following solutions for $\alpha(x,t)$
\begin{align}
\alpha(x, t)=\left\{\begin{array}{ll}
\alpha\left(x-v^{\star} t, 0\right), & 0\le t<\frac{x}{v^{\star}}, \\
\alpha\left(0, t-\frac{x}{v^{\star}}\right), &  t \geq \frac{x}{v^{\star}},
\end{array}\right.
\end{align}
and the following solutions for $\beta(x,t)$
\begin{align}
\beta\left(x, t\right)=\left\{\begin{array}{ll}
\beta\left(x+\varrho  t,0\right),  &0\le t<\frac{L-x}{\varrho  }, \\
\beta\left(L,t+\frac{x-L}{\varrho  }\right),  &t \geq \frac{L-x}{\varrho  }.
\end{array}\right.
\end{align}
Furthermore, by applying the boundary condition \eqref{targetbound1} to the target system \eqref{target1}-\eqref{target2}, it yields the following solution for $\alpha(x,t)$
\begin{align}\label{eq-sol-alpha}
\alpha(x, t)=\left\{\begin{array}{ll}
\alpha\left(x-v^{\star} t, 0\right), & 0\le t<\frac{x}{v^{\star}} ,\\
-k_{0} \beta\left(0, t-\frac{x}{v^{\star}}\right), &  t \geq \frac{x}{v^{\star}}.
\end{array}\right.
\end{align}
Similarly, by applying the boundary condition \eqref{targetbound}, it yields the following solution for $\beta(x,t)$
\begin{align}
\beta\left(x, t\right)=\left\{\begin{array}{ll}
\beta\left(x+\varrho  t,0\right),  &0\le t<\frac{L-x}{\varrho  }, \\
k \alpha\left(L, t+\frac{x-L}{\varrho  }\right)\notag \\
+z\left(0,t+\frac{x-L}{\varrho  }\right),  &t \geq \frac{L-x}{\varrho  }.
\end{array}\right. \label{beta}
\end{align}

Next, we aim to prove that  $\alpha(x,t)$ and $\beta(x,t)$ are both input-to-state stable. 
First, we estimate the bound of  $\|\alpha(\cdot, t)\|_{\infty}$. When $t<\frac{x}{v^{\star}}$,  the solution \eqref{eq-sol-alpha} yields that
\begin{align}
|\alpha(x, t)|=&\left|\alpha\left(x-v^{\star} t, 0\right)\right|\nonumber\\
\leq&   \| \alpha(\cdot, 0) \|_{L^\infty}\nonumber\\
\leq& e^{-\lambda (t-\frac{x}{v^\star})} \| \alpha(\cdot, 0) \|_{L^\infty}\nonumber\\
\leq& \varrho  _0e^{-\lambda t} \| \alpha(\cdot, 0) \|_{L^\infty},
\end{align}
where $\lambda>0$ and 
\begin{align}
 \varrho  _0=  e^{\lambda\frac{L}{v^\star}}.
\end{align}
When $t \geq \frac{x}{v^{\star}}$, the solution \eqref{eq-sol-alpha} yields that
\begin{align}
|\alpha(x, t)|&=\left|-k_{0} \beta\left(0, t-\frac{x}{v^{\star}}\right)\right| \notag \\
&\leq k_{0} \left\| \beta\left(\cdot, t-\frac{x}{v^{\star}}\right)\right\|_{L^\infty} \notag \\
&\leq k_{0} \sup _{0 \leq s \leq t}\| \beta(\cdot, s) \|_{L^\infty}.
\end{align}
From the characteristic formulas and the boundary conditions, there exist
\[
C_\alpha=e^{\frac{\mu L}{v^\star}},
\qquad
C_\beta=e^{\frac{\mu L}{\gamma p^\star-v^\star}},
\]
such that
\begin{align}
\|\alpha\|_{\mu,t}
&\leq
C_\alpha \|\alpha(\cdot,0)\|_{L^\infty}
+
C_\alpha |k_0|\|\beta\|_{\mu,t},
\label{alphaweighted}\\
\|\beta\|_{\mu,t}
&\leq
C_\beta \|\beta(\cdot,0)\|_{L^\infty}
+
C_\beta |k|\|\alpha\|_{\mu,t}
+
C_\beta \|z\|_{\mu,t}.
\label{betaweighted}
\end{align}
Choose \(\mu>0\) sufficiently small such that
\begin{align}
C_\alpha C_\beta |kk_0|<1.
\label{weightedgain}
\end{align}
This is possible under the small-gain condition \(|kk_0|<1\). Applying the
ISS small-gain theorem to \eqref{alphaweighted}--\eqref{betaweighted} gives
\begin{align}
\|\alpha\|_{\mu,t}+\|\beta\|_{\mu,t}
&\leq
C_\mu
(
\|\alpha(\cdot,0)\|_{L^\infty}
+
\|\beta(\cdot,0)\|_{L^\infty} \notag \\
&+
\|z\|_{\mu,t}
),
\label{abweighted}
\end{align}
for some \(C_\mu>0\).

Moreover, from \eqref{targetz} and \eqref{target2} under the controller
\eqref{controller1}, \(z\) satisfies
\[
\frac{D}{L}z_t=z_x,\qquad z(L,t)=0.
\]
Hence \(z(\cdot,t)\) vanishes after time \(D\), and
\begin{align}
\|z(\cdot,t)\|_{L^\infty}
\leq
e^{\mu D}e^{-\mu t}\|z(\cdot,0)\|_{L^\infty}.
\label{zexp}
\end{align}
Equivalently,
\begin{align}
\|z\|_{\mu,t}
\leq
e^{\mu D}\|z(\cdot,0)\|_{L^\infty}.
\end{align}
Combining this estimate with \eqref{abweighted}, we obtain
\begin{align}
&\|\alpha(\cdot,t)\|_{L^\infty}
+
\|\beta(\cdot,t)\|_{L^\infty}
+
\|z(\cdot,t)\|_{L^\infty}
\leq
\tilde C e^{-\mu t}\notag \\
&\left(
\|\alpha(\cdot,0)\|_{L^\infty}
+
\|\beta(\cdot,0)\|_{L^\infty}
+
\|z(\cdot,0)\|_{L^\infty}
\right),
\end{align}
for some \(\tilde C>0\).
The proof is complete.
\end{proof}

Ultimately, due to the equivalence between the closed-loop system \eqref{orin}-\eqref{output} and target system \eqref{target1}-\eqref{target2}, as shown in the previous subsection, the following main theorem is established.
\begin{theorem}
Consider the closed-loop system consisting of system dynamics \eqref{orin}-\eqref{orin1} and delay transport PDEs \eqref{u}-\eqref{ubond} with boundary conditions \eqref{bond1}-\eqref{output} and controller \eqref{controller1}. If $|kk_0|<1$, then for any initial conditions $\bar{w}(\cdot,0) ,\bar{v}(\cdot,0) ,u(\cdot,0) \in \mathcal{L}^{\infty}(0,L)$, the closed-loop system is exponentially stable in the sense that there exist positive constants $\tilde{c}$ and $\lambda$ such that 
\begin{align}
\Psi (t)\leq \tilde{c} e^{-\lambda t}\Psi (0),
\end{align}
where 
\begin{align}
\Psi (t)=\|\bar{w}(\cdot,t) \|_{L^\infty}+\|\bar{v}(\cdot,t) \|_{L^\infty}+\|u(\cdot,t) \|_{L^\infty}.
\end{align}
\end{theorem}

It shows the global exponential stability of the linearized model, while the implication for the nonlinear ARZ system is local.

\begin{remark}
To express the condition $|k k_0|<1$ in terms of physically interpretable traffic parameters, we start from the Greenshields relation
$p^\star = v_f\left(\frac{\rho^\star}{\rho_m}\right)^\gamma.$
It is convenient to introduce the dimensionless quantities
$r = \left(\frac{\rho^\star}{\rho_m}\right)^\gamma = \frac{p^\star}{v_f},
\alpha = \frac{L}{\tau v_f},$
where $r$ characterizes the normalized operating point and $\alpha$ represents the ratio between the controlled-road length and the relaxation length scale $\tau v_f$. Under these notations, the equilibrium quantities can be rewritten as
$p^\star = v_f r,
v^\star = v_f - p^\star = v_f(1-r).$
Substituting these identities into \eqref{k} and \eqref{k0}, respectively, yields
$k = \exp\!\left(-\frac{L}{\tau v^\star}\right)
  = \exp\!\left(-\frac{\alpha}{1-r}\right),
k_0 = \frac{\gamma p^\star-v^\star}{v^\star}
    = \frac{\gamma v_f r-v_f(1-r)}{v_f(1-r)}
    = \frac{(\gamma+1)r-1}{1-r}.$
Consequently, the small-gain condition can be expressed as
$
|k k_0|
=
\exp\!\left(-\frac{\alpha}{1-r}\right)
\left|\frac{(\gamma+1)r-1}{1-r}\right|.
$
In the congested regime, namely
$
r>\frac{1}{\gamma+1},
$
the above expression therefore simplifies to
$
|k k_0|
=
e^{-\alpha/(1-r)}
\frac{(\gamma+1)r-1}{1-r}.
$
This derivation makes explicit how the feasibility condition depends jointly on the operating point $r$, the sensitivity coefficient $\gamma$, and the dimensionless quantity $\alpha$.
Therefore, the condition $|k k_0|<1$ should be interpreted as a practical feasibility condition for representative congested freeway settings rather than a universal property.
\end{remark}

\section{Simulation examples} \label{4}

This section provides two simulations to validate the effectiveness of the proposed controller. First, we compare the proposed delay compensation controller and the controller without compensation in \cite{yu2019traffic} on the original nonlinear ARZ model with delayed boundary actuation. Second, we perform a data-informed nonlinear validation using NGSIM freeway trajectory data to calibrate the operating condition and model parameters.

\begin{figure}[htbp]
    \centering
    \begin{subfigure}[b]{0.8\linewidth}
        \includegraphics[width=\linewidth]{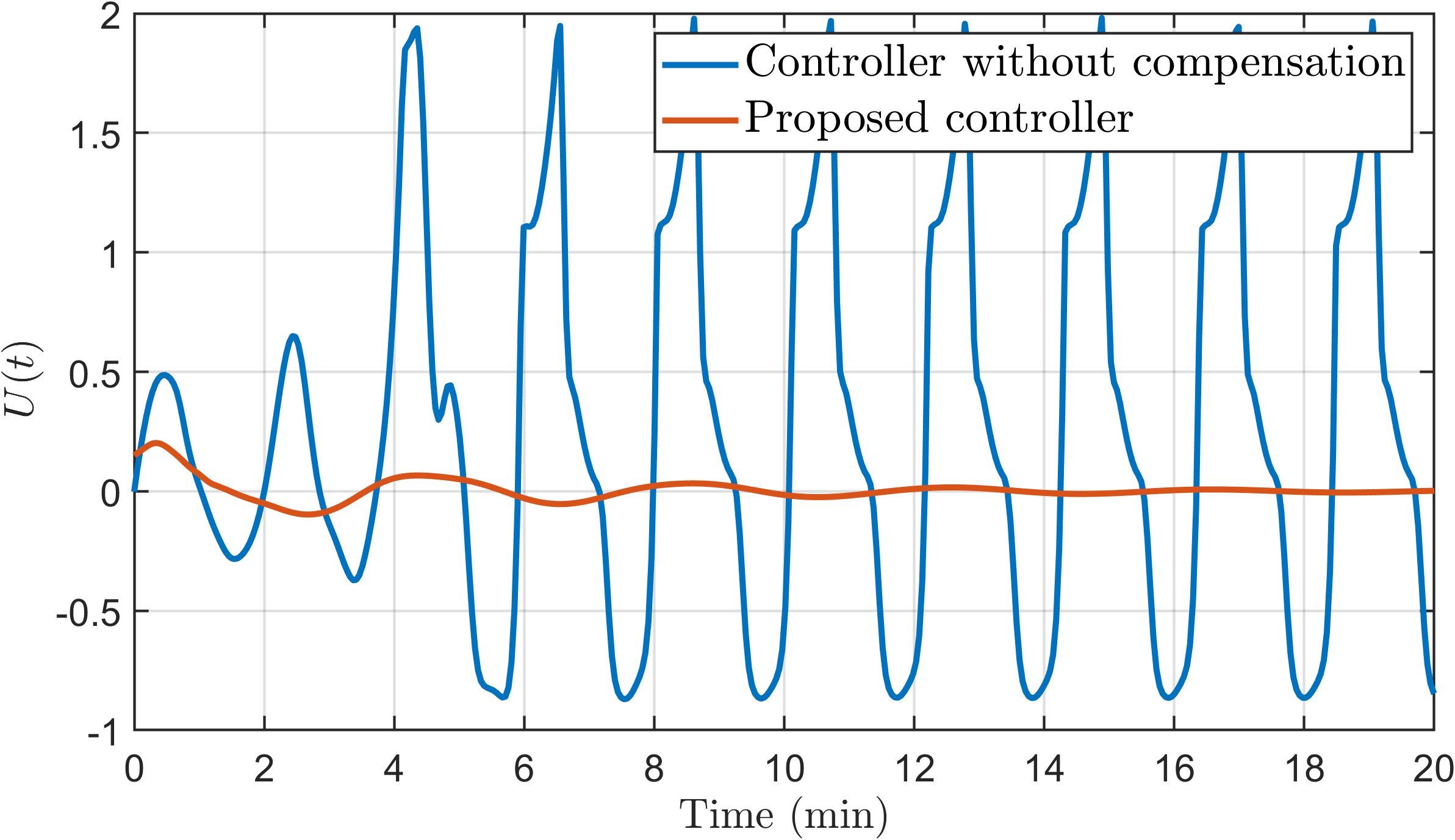}
        \caption{Evaluation of control input $U(t)$}
        \label{uc}
    \end{subfigure}
    \hfill
    \begin{subfigure}[b]{0.8\linewidth}
        \includegraphics[width=\linewidth]{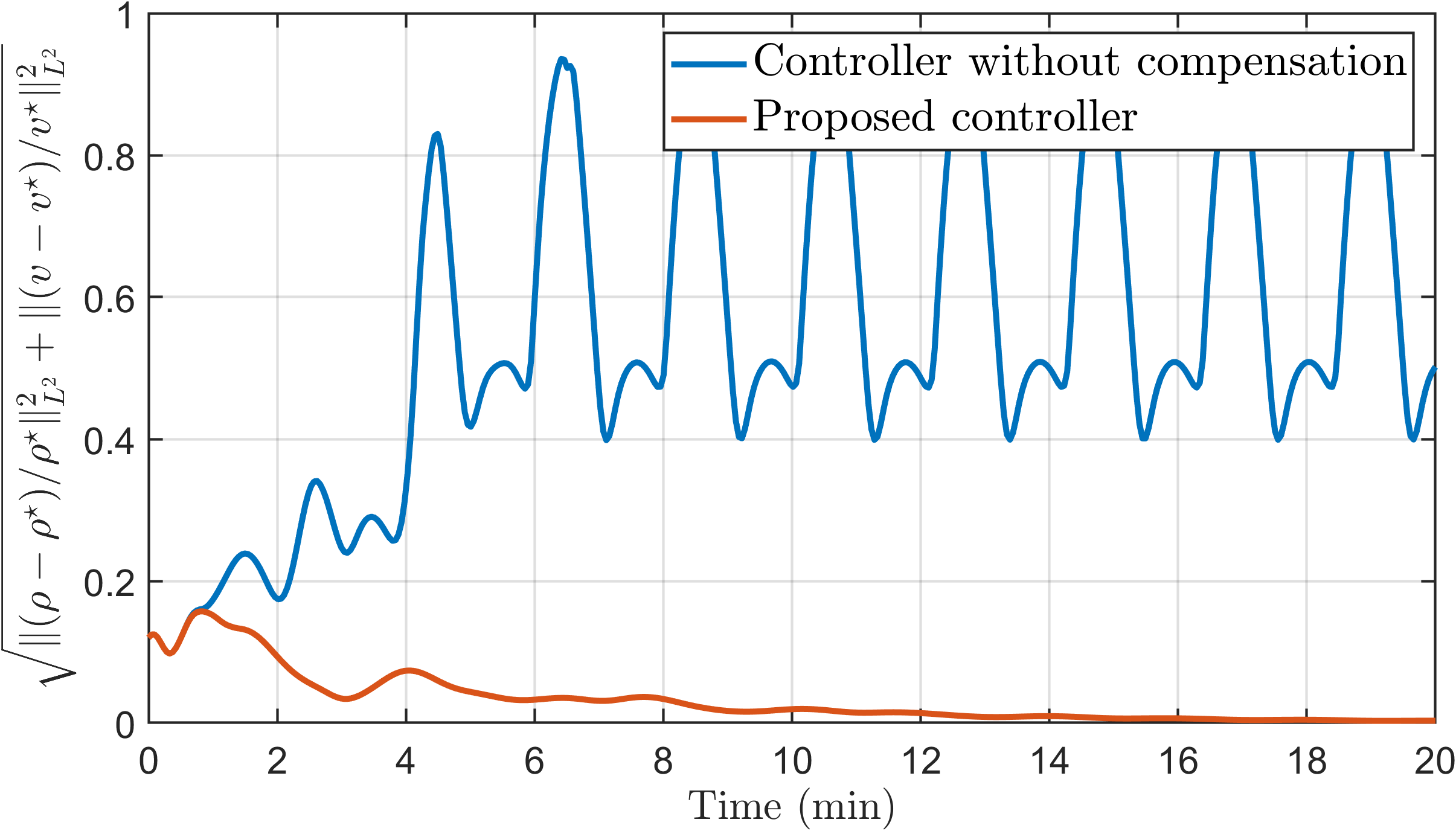}
        \caption{Evaluation of the state norm}
        \label{nc}
    \end{subfigure}
    \caption{Comparison of the controller without compensation \cite{yu2019traffic} and the proposed controller on the original nonlinear delayed ARZ model}
    \label{fig:comparison-normcontrol}
\end{figure}

\subsection{Comparison simulation}

We implement the ARZ model with the following parameters: $\gamma=1$, the length of freeway section is set to $L=1 \mathrm{~km}$, the free speed is  $\bar{v}_{f}=40 \mathrm{~m} / \mathrm{~s}$  and the maximum density is  $\rho_{~m}=150 \mathrm{~vehicles} / \mathrm{~km}$. The steady state  $\left(\rho^{\star}, v^{\star}\right)$  is chosen as $( 120 \mathrm{~vehicles} / \mathrm{~km}, 8 \mathrm{~m} / \mathrm{~s}  )$, which corresponds to the congested regime. The relaxation time is set to  $\tau=60 \mathrm{~s}$ and the input delay is set to $D=40\mathrm{~s}$. Sinusoidal initial conditions are employed.

\begin{figure}[htbp]
    \centering
    \begin{subfigure}[b]{0.8\linewidth}
        \includegraphics[width=\linewidth]{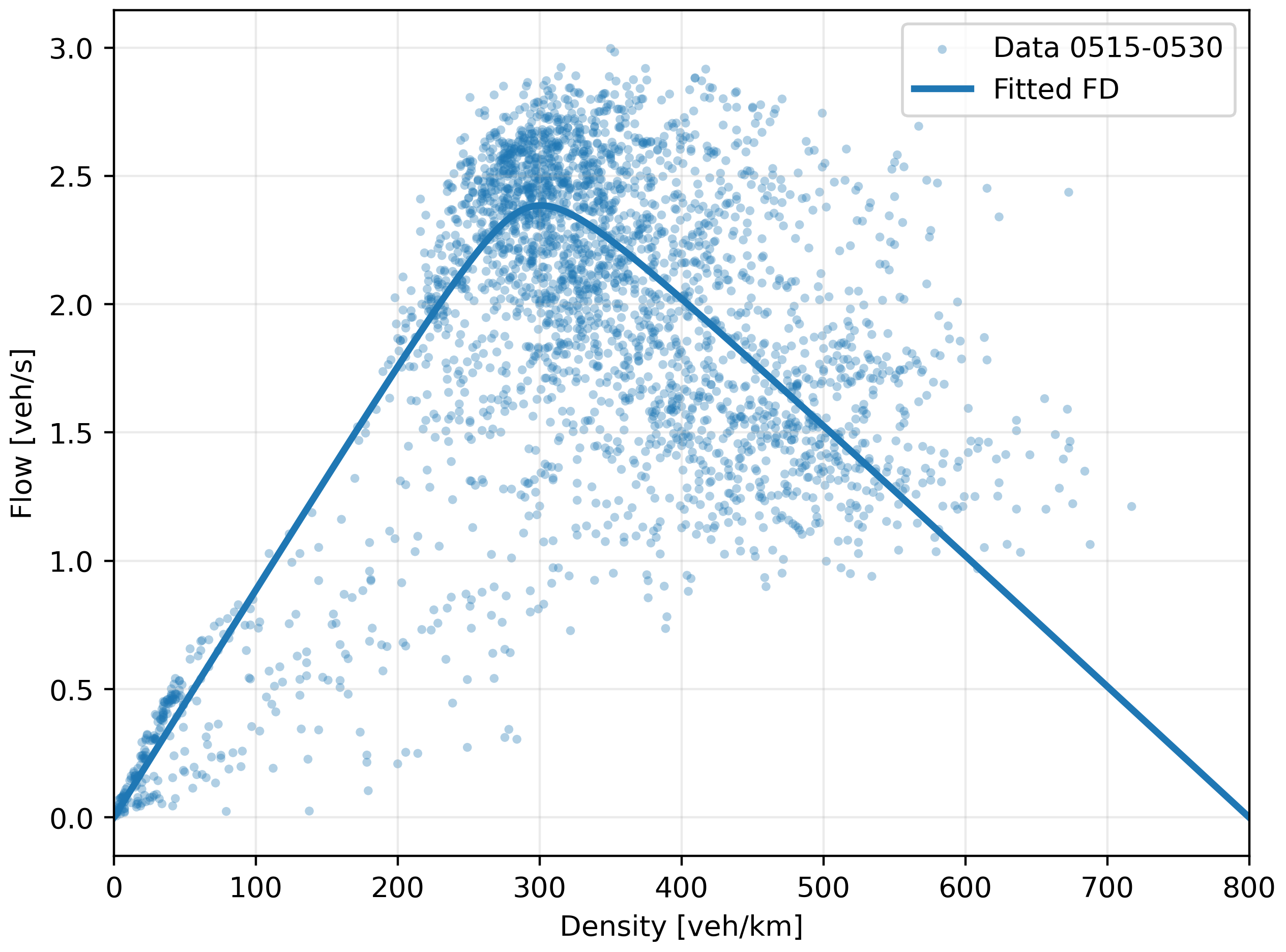}
        \caption{Density-flow scatter and fitted fundamental diagram }
        \label{ngsim}
    \end{subfigure}
    \hfill
    \begin{subfigure}[b]{0.8\linewidth}
        \includegraphics[width=\linewidth]{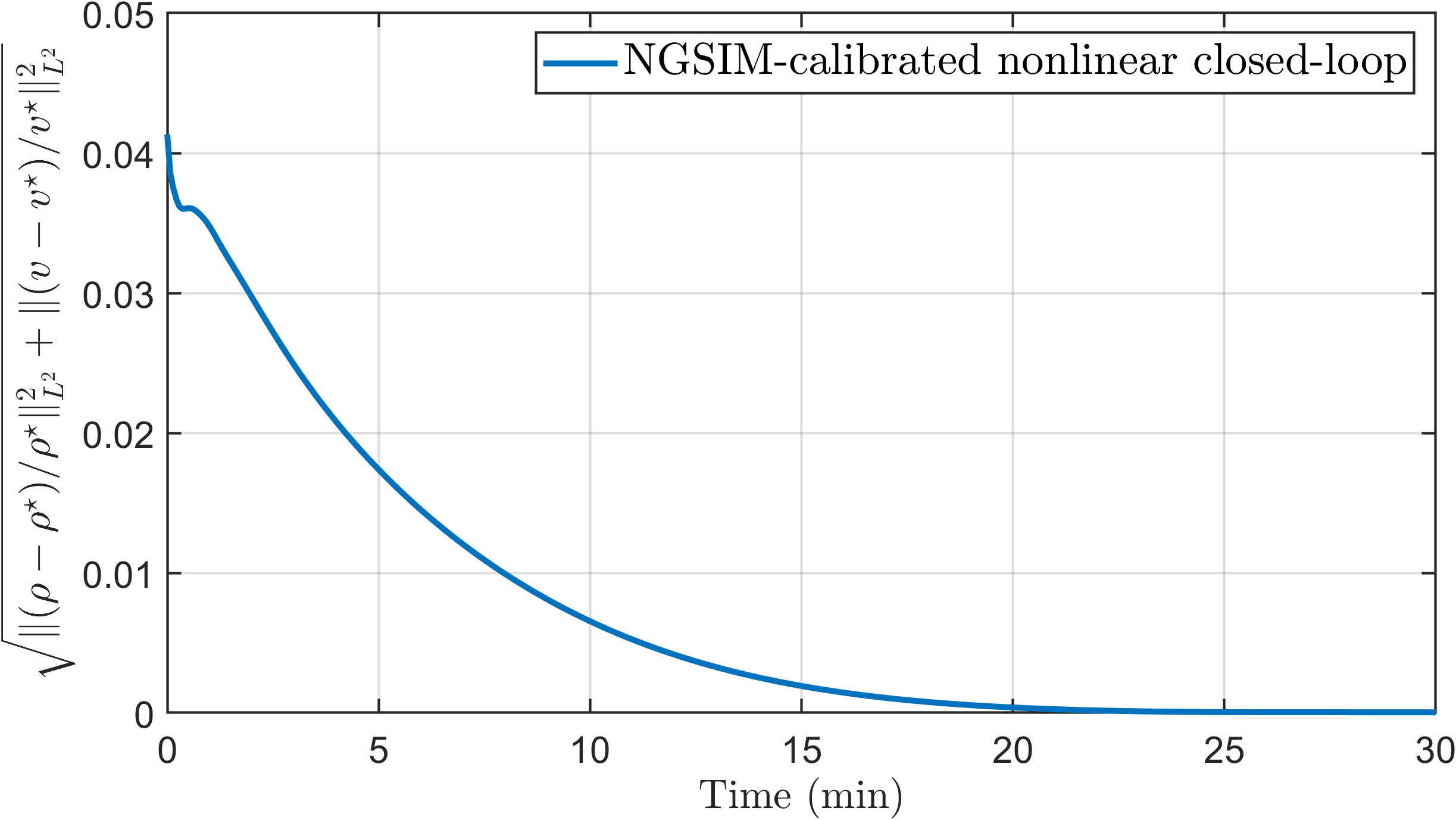}
        \caption{Evaluation of the state norm}
        \label{ngn}
    \end{subfigure}
    \caption{NGSIM validation of the proposed controller on the original nonlinear delayed ARZ model}
    \label{fig:ngsim}
\end{figure}

 We conducted a comparative experiment using two controllers. The first is the proposed controller. The second is the controller without compensation proposed in \cite{yu2019traffic}, serving as a delay-uncompensated benchmark. 
 As shown in Figure \ref{fig:comparison-normcontrol}(\subref{uc}) and Figure \ref{fig:comparison-normcontrol}(\subref{nc}), the controller without compensation yields pronounced oscillations, whereas the proposed controller suppresses the transient much more rapidly and drives the nonlinear state to a small neighborhood of the equilibrium within a  short time. This comparison shows that explicit compensation of the input delay is essential when the delayed actuation is applied to the traffic-flow-control setting.  This demonstrates the necessity of explicit delay compensation and indicates that the proposed controller can effectively mitigate congestion in the considered delayed ARZ setting and drive the state toward the desired equilibrium.

\subsection{Validation on freeway data}
The widely recognized Next Generation Simulation (NGSIM) dataset \cite{dot2018next} is adopted, focusing on vehicle trajectory records collected from a segment of Interstate 80 (I-80) in Emeryville, California, on April 13, 2005. This dataset offers microscopic positions and velocities of individual vehicles at a temporal resolution of 0.1 seconds. Trajectories extracted during the evening peak period (5:15 PM-5:30 PM) are used, which feature pronounced stop‑and‑go traffic oscillations. The macroscopic physical domain is defined with a length of $L=500~\mathrm{m}$. Given that the freeway section consists of six lanes, the maximum achievable physical density is strictly capped at $\rho_m=800~\mathrm{veh/km}$.

The trajectory data are first aggregated to reconstruct macroscopic traffic variables based on the calibration method in \cite{yu2020pde} and then used to calibrate the equilibrium traffic relation and the operating condition of the ARZ model. The delay time is set as $D=6 \mathrm{~s}$. Figure \ref{fig:ngsim}(\subref{ngsim}) shows the resulting density-flow scatter together with the fitted fundamental diagram. The proposed controller is then applied to the calibrated original nonlinear delayed ARZ model. Figure \ref{fig:ngsim}(\subref{ngn}) shows the evolution of the corresponding state norm. The norm decreases rapidly after the initial transient, indicating that the controller remains effective under data-informed parameter values and realistic congestion patterns. The simulation results show that the original nonlinear system is locally stable under the proposed controller. 

\section{Conclusion} \label{5}
This paper proposes a backstepping-based controller for traffic congestion control using the ARZ model, considering arbitrarily large input delays. The ARZ model characterizes a $2 \times 2$ hyperbolic PDE system with proximal boundary reflection, which presents substantial challenges when coupled with input delay. 
To address these challenges, we developed a control framework comprising three key components. First, we transform the arbitrarily large input delay into a transport PDE. This transformation converts the original  $2 \times 2$ hyperbolic PDE system with arbitrarily large input delay into an equivalent $3 \times 3$ hyperbolic PDE system without delay. 
Second, a novel backstepping transformation and target system are designed. 
A technical contribution of this work is the well-posedness analysis of the kernel equations associated with the delay-compensating transformation. For hyperbolic PDEs with delays, the kernel function is subject to two boundary conditions. Consequently, the well-posedness analysis should be performed separately from the characteristic line to each of the two boundaries.
We apply the ISS small-gain theorem for stability analysis. It handles the strongly coupled terms from proximal reflection more effectively than conventional Lyapunov-based analysis. 
Finally, two validation simulations are provided. The first one compares the proposed controller with the controller without compensation on the original nonlinear delayed ARZ model, showing a clear advantage of delay compensation. The second simulation validates the proposed controller using real traffic data. These simulations are performed on the original nonlinear system, demonstrating the effectiveness of the proposed controller.
Future work will consider observer-based output-feedback controller design under limited measurements, together with the associated estimation-error analysis. Future work will also investigate robustness and disturbance-rejection extensions in the presence of distributed and boundary disturbances, as well as measurement noise.

\section{Disclosure statement}
No potential conflict of interest was reported by the author(s).

\section{Funding}
This work was supported by the National Natural Science Foundation of China under Grants 62373314 and the Research Grants Council of the Hong Kong Special Administrative Region of China under Grant CityU/11210424.

\section{Data availability statement}
The authors confirm that the data supporting the findings of this study are available within the article or its supplementary materials.

 \bibliographystyle{elsarticle-num}
 \bibliography{Refs}
 
\end{document}